\documentclass[10pt]{amsart}

\usepackage{amsmath,amsfonts,amssymb}
\usepackage{graphicx,verbatim}
\usepackage[numbers,square,sort&compress]{natbib}
\usepackage[colorinlistoftodos]{todonotes}
\usepackage{times}
\usepackage{xspace}
\usepackage{xypic}
\usepackage{color}
\usepackage{a4wide}
\usepackage{tikz}
\usetikzlibrary{arrows}
\usetikzlibrary{decorations.markings}
\usepackage{subfig}
\usepackage{enumerate}

\theoremstyle{plain}
  \newtheorem{thm}{Theorem}[section]
  \newtheorem{prop}[thm]{Proposition}
  \newtheorem{defn-prop}[thm]{Definition-Proposition}
  \newtheorem{lem}[thm]{Lemma}
  \newtheorem{cor}[thm]{Corollary}

\theoremstyle{definition}
  \newtheorem{defn}[thm]{Definition}
  \newtheorem{ex}[thm]{Example}
  \newtheorem{rem}[thm]{Remark}

\def\F{\mathbb{F}}

\def\SSS{\mathfrak{S}}

\begin{document}

\title{Stratification and Enumeration of Boolean Functions by Canalizing Depth}

\author[Q.~He]{Qijun He} \address{Department of
  Mathematical Sciences \\ Clemson University \\ Clemson, SC 29634-0975, USA}
\email{qhe@g.clemson.edu, macaule@clemson.edu}

\author[M.~Macauley]{Matthew Macauley} 

\thanks{Partially supported by NSF grant DMS-1211691.}



\begin{abstract}
Boolean network models have gained popularity in computational systems biology over the last dozen years. Many of these networks use canalizing Boolean functions, which has led to increased interest in the study of these functions. The canalizing depth of a function describes how many canalizing variables can be recursively ``picked off'', until a non-canalizing function remains. In this paper, we show how every Boolean function has a unique algebraic form involving extended monomial layers and a well-defined core polynomial. This generalizes recent work on the algebraic structure of nested canalizing functions, and it yields a stratification of \emph{all} Boolean functions by their canalizing depth. As a result, we obtain closed formulas for the number of $n$-variable Boolean functions with depth $k$, which simultaneously generalizes enumeration formulas for canalizing, and nested canalizing functions. 
\end{abstract}

\maketitle

\section{Introduction}

Boolean networks were invented in 1969 by S.~Kauffman, who proposed them as models of gene regulatory networks \cite{kauffman69metabolic}. They were slow to catch on, but since a seminal paper \cite{albert2003topology} from 2003, where Albert and Othmer modeled the segment polarity gene in the fruit fly \emph{Drosophila melanogaster}, they have emerged as popular models for a variety of biological networks. Random Boolean networks (RBNs) have been studied throughout the years, with various restrictions on the functions or wiring diagrams to better reflect salient properties of actual biological networks. For example, without such restrictions, RBNs display chaotic behavior in the sense that they are very sensitive to small perturbations. In contrast, biological systems must be robustly designed \cite{li2004yeast} in order to withstand a variety of internal (e.g., mutation or gene knockout) and external (e.g., environmental) changes. In 1942, the geneticist H.~Waddington defined the concept of \emph{canalization} to study this robustness. Over 30 years later in \cite{kauffman1974large}, Kauffman introduced the notion of canalizing Boolean functions in order to accurately reflect the behavior of biological systems in the setting of Boolean network models. Another thirty years after that, Kauffman and collaborators further expanded the canalization concept and introduced the class of nested canalizing functions \cite{kauffman2003random}, which can be thought of as functions that are fully ``recursively canalizing.''

In the last decade, canalizing functions have been extensively studied by researchers in the fields of mathematics, biology, physics, computer science, and electrical engineering. For example, Shmulevich and Kauffman showed that canalizing functions have lower activities and sensitivities than random Boolean functions, and this causes Boolean network models using these functions to be more stable; see \cite{shmulevich2004activities} and \cite{kauffman2004genetic}. More work on the dynamical stability of canalizing Boolean networks was done in \cite{moreira2005canalizing} and in \cite{karlsson2007order}, where the authors explored the relationship between the proportion of canalizing functions in a network, and whether it lies in the ordered or chaotic dynamical regime, or near the so-called critical threshold. The evolution of canalizing Boolean networks was studied in \cite{szejka2007evolution}. Fourier analysis has shown that canalizing Boolean networks maximize mutual information \cite{klotz2014canalizing}. An exact formula was derived for the number of Boolean canalizing functions in \cite{just2004number}. Canalizing functions have been generalized from Boolean to over general finite fields in \cite{murrugarra2012number}.

Nested canalizing functions (NCFs) have also gained significant attention. In \cite{peixoto2010phase} and \cite{kadelka2014nested}, the authors study the phase diagram of Boolean networks with NCFs. A recursive formula for the number of NCFs was derived in \cite{jarrah2007nested}, where they were shown to be what the electrical engineering community calls unate cascade functions \cite{bender1978asymptotic}. NCFs have been studied algebraically through the lens of toric varieties \cite{jarrah2007discrete}, and in \cite{li2013boolean}, where the authors obtained a unique algebraic form by writing an NCF in extended monomial layers. This allowed the authors to enumerate the number of NCFs. It also provided the tools for the development of an algorithm in \cite{hinkelmann2012inferring} to reverse-engineering a nested canalizing Boolean network from partial data. In \cite{layne2012nested}, the authors generalized the notion of both canalizing and nested canalizing functions by introducing the class of partially nested canalizing functions. Loosely speaking, these are the functions that are ``somewhat recursively canalizing.'' The dynamics of Boolean networks built with these functions has been studied in \cite{layne2012nested} and \cite{jansen2013phase}.

In this article, we carry out a detailed mathematical study on canalization of Boolean functions. Instead of thinking of partially (or fully) nested canalizing functions as a subclass of Boolean functions, we consider canalization as a property of \emph{all} Boolean functions. We modify the notion of \emph{canalizing depth} from \cite{layne2012nested} to quantify the degree to which a function exhibits a recursive canalizing structure. From here, we show that every Boolean function has a unique algebraic form using extended monomial layers, generalizing what was done for NCFs in \cite{li2013boolean}. Once one ``peels off'' these layers, a unique non-canalizing \emph{core polynomial} remains. This gives a well-defined stratification of \emph{all} Boolean functions by canalizing depth and monomial layers, which includes the canalizing, non-canalizing, and NCFs as special cases. We say that a function is \emph{$k$-canalizing} if it has canalizing depth at least $k$. Our stratification allows us to derive exact formulas for the number the $k$-canalizing functions on $n$ variables. The special cases of $k=1$ and $k=n$ yield the enumeration results of canalizing, and nested canalizing functions from \cite{just2004number} and \cite{li2013boolean}, respectively. 

This paper is organized as follows. After introducing necessary preliminaries in Section~\ref{sec:prelims}, we define $k$-canalizing functions, canalizing depth and core functions in Section~\ref{sec:k-canalizing}. Next, we characterize Boolean functions by a unique polynomial form in Section~\ref{sec:layers} and use this to stratify all Boolean functions by extended monomial layers and their core polynomials, which are slighly different from the aforementioned core functions. In Section~\ref{sec:enumeration}, we use this structure to derive exact enumeration formulas for the number of functions with a fixed canalizing depth. Finally, we end in Section~\ref{sec:conclusions} with some concluding remarks and directions of current and future research.

\section{Canalizing and nested canalizing functions}\label{sec:prelims}

To make this paper self-contained we will restate some well-known definitions; see, e.g., \cite{kauffman2003random}. This is also needed because there are slight variations in certain definitions throughout the literature. Let $\F_2=\{0,1\}$ be the binary field, and let $f\colon\F_2^n\to\F_2$ be an $n$-variable Boolean function. 
\begin{defn}
A Boolean function $f(x_1,\ldots,x_n)$ is \emph{essential} in the variable $x_i$ if there exists a sequence $a_1,\ldots,a_{i-1},a_{i+1},\ldots,a_n\in \F_2$ such that 
\[
f(a_1,\ldots,a_{i-1},0,a_{i+1},\ldots,a_n)\neq f(a_1,\ldots,a_{i-1},1,a_{i+1},\ldots,a_n).
\]
In this case, we say that $x_i$ is an \emph{essential variable} of $f$. Variables that are non-essential are \emph{fictitious}.
\end{defn}
S.~Kauffman defined canalizing Boolean functions in \cite{kauffman1974large} to capture the general stability of gene regulatory networks. In that paper, a Boolean function $f$ is canalizing in variable $x_i$, with canalizing input $a$ and canalized output $b$, if, whenever $x_i$ takes on the value $a$, the output of $f$ is $b$, regardless of the inputs of other variables. As a consequence, constant functions are trivially canalizing. We will soon see why it is more mathematically natural to exclude these functions, among others. This is done by the following small adjustment to the original definition that does not change the overall idea. 
\begin{defn}\label{defn:canalizing}
A Boolean function $f\colon\F_2^n\to\F_2$ is \emph{canalizing} if there exists a variable $x_i$, a Boolean function $g(x_1,\ldots,x_{i-1},x_{i+1},\ldots,x_n)$, and $a,b\in\F_2$ such that
\[
f(x_1,\ldots,x_n)=\begin{cases}
b &  x_i=a,\\
g\not\equiv b &  x_i\neq a.
\end{cases}
\]
In this case, $x_i$ is a \emph{canalizing variable}, the input $a$ is the \emph{canalizing input}, and the output value $b$ when $x_i=a$ is the corresponding \emph{canalized output}.
\end{defn}
The only difference of our definition is the added restriction that $g$ can not be the constant function $b$. In other words, \emph{we require a canalizing function to be essential in its canalizing variable}. The original definition was motivated by the stability of canalizing functions while our definition tries to capture the dominance of the canalizing variable. At first glance, our additional restriction might seem artificial or insignificant. However, it is unequivocally more natural when considering the algebraic structure of Boolean functions, which is at the heart of the stratification derived in this paper.

In Definition~\ref{defn:canalizing}, when the canalizing variable does not receive its canalizing input $a$, the function $g$ obtained by plugging in $x_i=\overline{a}$ can be an arbitrary Boolean function. To better model a dynamically stable network, in \cite{kauffman2003random} Kauffman proposed that in this case, there should be another variable $x_j$ that is canalizing for a particular input, and so on. This leads to the following definition, where $\sigma$ is a total ordering, or permutation, of $[n]:=\{1,\dots,n\}$. We write this as $\sigma=\sigma(1),\sigma(2),\dots,\sigma(n)$, and say that $\sigma\in\SSS_n$, the symmetric group on $[n]$.
\begin{defn}
A Boolean function $f\colon\F_2^n\to\F_2$ is \emph{nested canalizing} with respect to the permutation $\sigma\in\SSS_n$, inputs $a_i$ and outputs $b_i$, for $i=1,2,\ldots,n$, if it can be represented in the form:
\begin{equation}\label{eqn:ncf}
f(x_1,\ldots,x_n)=\begin{cases}
b_1 &   x_{\sigma(1)}=a_1,\\
b_2 &   x_{\sigma(1)}\neq a_1,x_{\sigma(2)}=a_2,\\
b_3 &   x_{\sigma(1)}\neq a_1,x_{\sigma(2)}\neq a_2,x_{\sigma(3)}=a_3,\\
\vdots & \vdots \\
b_n &   x_{\sigma(1)}\neq a_1,\ldots, x_{\sigma(n-1)}\neq a_{n-1}, x_{\sigma(n)}=a_n,\\
\overline{b_n} &   x_{\sigma(1)}\neq a_1,\ldots, x_{\sigma(n-1)}\neq a_{n-1}, x_{\sigma(n)}\neq a_n.
\end{cases}
\end{equation}
\end{defn}
The idea of nested canalizing is in that some sense, it is ``recursively canalizing'' for exactly $n$ steps. As an analogy, one can consider a nested canalizing function as an onion. We can peel off variables one at a time by not taking the canalizing input of each variable (i.e., by plugging in $x_i=\overline{a_i}$). Before we peel off the `inner' variables, we need to peel off the `outer' variables first. In the end, we are left with the constant function $\overline{b_n}$. We will return to this onion analogy several times throughout this paper to highlight our main ideas.
\begin{rem}\label{rem:essential-ncf}
Since $b_n\neq \overline{b_n}$, a nested canalizing function is essential in all $n$ variables.
\end{rem}
If a Boolean function is nested canalizing, then at least one (of all $n!$) ordering of the variables yields an equation in the form of Eq.~\eqref{eqn:ncf}. Note that such variable orderings are not unique, and the number of such orderings depends on the function $f$. For example, we can write the function $f_1(x,y,z)=xyz$ as in Eq.~\eqref{eqn:ncf} using any of the $6$ orderings of the variables $\{x,y,z\}$. In contrast, for $f_2(x,y,z)=x(yz+1)$, only $2$ orderings would work, namely $(x,y,z)$ and $(x,z,y)$.

\section{$k$-Canalizing Functions}\label{sec:k-canalizing}

Nested canalizing functions have a very restrictive structure and become increasingly sparse as the number of input variables increases \cite{jarrah2007nested}. In a real network model, it is often the case that not all variables exhibit nested canalizing behavior. Moreover, the first several canalizing variables play more central roles than the remaining variables. Thus, it is natural to consider functions that are canalizing, but not nested canalization. For example, one function in the segment polarity gene in by Albert and Othmer's seminal paper \cite{albert2003topology} is canalizing but not nested canalizing. For another example, one can look at the lactose (\emph{lac}) operon, which regulates the transport and metabolism of lactose in \emph{Escherichia coli}. In \cite{robeva2013mathematical}, a simple Boolean network model of the \emph{lac} operon was proposed, where the regulatory function for lactose was
\[
f_L(t+1)=\overline{G_e}\wedge[(L\wedge\overline{E})\vee(L_e\wedge E)]\,.
\]
In a sentence, this means ``internal lactose ($L$) will be present the following timestep if there is no external glucose ($G_e$), \emph{and} at least one of the following holds:
\begin{itemize}
\item there already is internal lactose present, but the enzyme $\beta$-galactosidase ($E$) that breaks it down is absent;
\item there is external lactose $(L_e)$ available and the \emph{lac} permease transporter protein (also represented by $E$ since it is transcribed by the same gene) is present.
\end{itemize}
The variable $\overline{G_e}$ (though sometimes considered a parameter) is canalizing because it acts as a ``shut-down'' switch: if $G_e=1$, then $f_L=0$ regardless of the other variables. In other words, we can write this as 
\[
f_L(G_e,L_e,L,E)
=\begin{cases}
0 &  G_e=1,\\
(L\wedge \overline{E})\vee(L_e\wedge E) &  G_e\neq 0.
\end{cases}
\]
The function $g=(L_e\wedge E)\vee(L\wedge \overline{E})$ is not canalizing, and so the $5$-variable function $f_L$ is canalizing but not nested canalizing. In the framework that we are about to define, this function has canalizing depth $1$. 

Due to both theoretical and practical reasons, a relaxation of the nested canalizing structure is often necessary. This was done in \cite{layne2012nested}, where there authors defined partially nested canalizing functions, and then distinguished between the ``active depth'' and ``full depth'' of a function. Our definition of $k$-canalizing functions is similar to what it means in their paper to be ``partially nested canalizing of active depth at least $k$.'' As before, the small differences are motivated by the desire to have a natural unique algebraic form. 
\begin{defn}\label{defn:k-canalizing}
A Boolean function $f(x_1,\ldots,x_n)$ is \emph{k-canalizing}, where $0\leq k\leq n$, with respect to the permutation $\sigma\in\SSS_n$, inputs $a_i$,
and outputs $b_i$, for $1\leq i \leq k$, if
\begin{equation}\label{eqn:pncf}
f(x_1,\ldots,x_n)=\begin{cases}
b_1 &   x_{\sigma(1)}=a_1,\\
b_2 &   x_{\sigma(1)}\neq a_1,x_{\sigma(2)}=a_2,\\
b_3 &   x_{\sigma(1)}\neq a_1,x_{\sigma(2)}\neq a_2,x_{\sigma(3)}=a_3,\\
\vdots & \vdots \\
b_k &   x_{\sigma(1)}\neq a_1,\ldots, x_{\sigma(k-1)}\neq a_{k-1}, x_{\sigma(k)}=a_k,\\
g\not\equiv b_k &   x_{\sigma(1)}\neq a_1,\ldots, x_{\sigma(k-1)}\neq a_{k-1}, x_{\sigma(k)}\neq a_k.
\end{cases}
\end{equation}
where $g=g(x_{\sigma(k+1)},\ldots,x_{\sigma(n)})$ is a Boolean function on $n-k$ variables. When $g$ is not a canalizing function, the integer $k$ is the \emph{canalizing depth} of $f$. Furthermore, if $g$ is not a constant function, then we call it a \emph{core function} of $f$, denoted by $f_C$. 
\end{defn}
As with canalizing and nested canalizing functions, the $g\not\equiv b_k$ condition ensures that $f$ is essential in the final variable, $x_{\sigma(k)}$. 

\begin{rem}\label{rem:essential-pncf}
Since $g\not\equiv b_k$, a function $f$ that is $k$-canalizing with respect to $\sigma\in\SSS_n$, inputs $a_i$ and outputs $b_i$ is essential in each $x_{\sigma(i)}$ for $i=1,\dots,k$.
\end{rem}

The representation of a $k$-canalizing function $f$ in the form of Eq.~\eqref{eqn:pncf}, even when $k$ is the canalizing depth, is generally not unique since it depends on the variable ordering. However, we will prove that several key properties, such as the canalizing depth and core function $f_C=g$ (if there is one), are independent of representation. It is worth noting that if $g$ is constant, then $g$ need not be unique, i.e., both $g\equiv 0$ and $g\equiv 1$ can arise. This is why we do not allow constant core functions. The following observation is elementary. 

\begin{rem}\label{rem:swap}
 If $f$ is $k$-canalizing with respect to $\sigma\in\SSS_n$, inputs $a_i$ and outputs $b_i$, then any initial segment $x_{\sigma(1)},\dots,x_{\sigma(j)}$ with the same canalized output $b_1=\cdots=b_j$ can be permuted to yield an equivalent form as in Eq.~\eqref{eqn:pncf}. 
\end{rem}

\begin{defn}\label{defn:g-function}
If $f(x_1,\dots,x_n)$ is $k$-canalizing with respect to $\sigma\in\SSS_n$, inputs $a_i$ and outputs $b_i$, then for each $j\leq k$, define the Boolean function $g^\sigma_j(x_{\sigma(j+1)},\dots,x_{\sigma(n)})$ to be the result of plugging in $x_{\sigma(i)}=\overline{a_i}$ for $i=1,\dots,j$.
\end{defn}

In plain English, the function $g^\sigma_j$ is the result of when the first $j$ canalizing variables do \emph{not} get their canalizing inputs. We can now show that the canalizing depth $k$ and the core function $f_C$ are independent of the order of the variables. Moreover, the ambiguity of variable orderings is well-controlled in that they are partitioned into blocks called \emph{layers} via extended monomials, and variables can be permuted arbitrarily if and only if they lie in the same layer. This generalizes the observation in Remark~\ref{rem:swap}.

\begin{prop}\label{prop:well-defined}
  Suppose an $n$-variable Boolean function $f$ is $k$-canalizing with respect to the permutation $\sigma$, inputs $a_i$ and outputs $b_i$, for $1\leq i \leq k$, and $k'$-canalizing with respect to the permutation $\sigma'$, inputs $a'_j$ and outputs $b'_j$, for $1\leq j \leq k'$, such that both $g$ and $g'$, obtained by substituting $\overline{a_i}$ for $x_{\sigma(i)}$ and $\overline{a_j'}$ for $x_{\sigma'(j)}$ respectively, are not canalizing. Then $k=k'$ and the resulting core functions, if they exist, are the same.
\end{prop}

\begin{proof}
Assume $f$ is canalizing, because otherwise, $k=k'=0$ and the result is trivial. Without losing generality we can assume $\sigma(1)\neq \sigma'(1)$, since if this were not the case, we could simply input $\overline{a_1}=\overline{a_1'}$ for $x_{\sigma(1)}=x_{\sigma'(1)}$ and consider $g_1^\sigma=g_1^{\sigma'}$. (Note that if $\sigma(1)=\sigma'(1)$ and $a_1\neq a_1'$, then $b_1\neq b_1'$, which means that $f$ is completely determined by the input to $x_{\sigma(1)}=x_{\sigma'(1)}$. In this case, $f$ has only one essential variable, and so $k=1$. Moreover, both $g_1^\sigma$ and $g_1^{\sigma'}$ are constant functions. Thus $f$ has no core function.)

Since $g$ is non-canalizing, it is not essential in $x_{\sigma(1)}$, and thus $\sigma(1)=\sigma'(j^*)$ for some $1<j^*\leq k'$. We claim that we may assume without loss of generality that $a_{j^*}'=a_1$ and $b_{j^*}'=b_1$. To see why, first suppose that $a_{j^*}'=\overline{a_1}$ and consider the two possible inputs to $x_{\sigma'(j^*)}=x_{\sigma(1)}$ in the function $g_{j^*-1}^{\sigma'}$. If this variable takes its canalizing input $\overline{a_1}$, then the output is $b_{j^*}'$. However, since $f$ is canalizing in $x_{\sigma'(j^*)}=x_{\sigma(1)}$, then the other input $a_1$ would yield the output $b_1$. In other words, $g_{j^*-1}^{\sigma'}$ is completely determined by the input to $x_{\sigma'(j^*)}$, so all subsequent variables are fictitious. Therefore, $g_{j^*}^{\sigma'}=g'$ must be constant, hence $j^*=k'$. Moreover, this function must be $g'\equiv b_1$ because it only arises when $x_{\sigma'(j^*)}=x_{\sigma(1)}$ takes the canalizing input $a_1$. Since $f$ is essential in $x_{\sigma'(j^*)}=x_{\sigma(1)}$, then Remark~\ref{rem:essential-pncf} implies that $b'_{j^*}=\overline{b_1}$, the opposite value of $g'\equiv b_1$. Thus, we have two equivalent ways to represent  $g_{j^*-1}^{\sigma'}=g_{k'-1}^{\sigma'}$:
\begin{equation}\label{eqn:swap}
  g_{k'-1}^{\sigma'}=\begin{cases}
  \overline{b_1} & x_{\sigma'(k')}=\overline{a_1},\\
  g'\equiv b_1  & x_{\sigma'(k')}=a_1.
  \end{cases}
  \quad=\begin{cases}
  b_1 & x_{\sigma'(k')}=a_1,\\
  g'\equiv \overline{b_1}  & x_{\sigma'(k')}=\overline{a_1}.
  \end{cases}
\end{equation}
In other words, switching the triple of values $(a'_{k'},b'_{k'},g')$ from $(\overline{a_1},\overline{b_1},b_1)$ to $(a_1,b_1,\overline{b_1})$ in the original representation of $f$ with respect to $\sigma'\in\SSS_n$ does not change the function, so we may assume that $a'_{j^*}=a_1$ and $b'_{j^*}=b_1$, as claimed. The proof for the case when $b_{j^*}'=\overline{b_1}$ is almost the same.

Since $f$ is canalizing in $x_{\sigma'(j^*)}=x_{\sigma(1)}$ with input $a_1$ and output $b_1$, we must also have $b_j'=b_1$ for all $1\leq j\leq j^*$. By Remark~\ref{rem:swap}, we can create a new permutation $\sigma''$ by swapping the order of $x_{\sigma'(1)}$ and $x_{\sigma'(j^*)}$ in $\sigma'$. Clearly, $f$ is $k'$-canalizing with respect to $\sigma''$ and $g_{k'}^{\sigma'}=g_{k'}^{\sigma''}$. Since $x_{\sigma(1)}=x_{\sigma''(1)}$, the result follows from induction on $g_1^\sigma=g_1^{\sigma''}$. We conclude that $k=k'$. 

Finally, we need to show that when $f$ has a core function $f_C$, it is unique. The non-canalizing functions $g$ and $g'$ are essential in the same set of variables. If they are both constant functions, then they actually need not be the same, due to the different ways to write $g'$ as in Eq.~\eqref{eqn:swap}. Otherwise, they are core functions for $f$, and are obtained by substituting the same set inputs for the same set of variables, thus we must have $f_C=g=g'$.
\end{proof}

It is worth noting that Definition \ref{defn:k-canalizing} is similar to the definition of $k$-partially nested canalizing functions ($k$-PNCFs) in \cite{layne2012nested}. In fact, these two definitions hold the same motivation but  are from different perspectives. In \cite{layne2012nested}, the authors treat $k$-PNCFs as a subclass of Boolean functions. While we prefer to consider canalization as a property of Boolean functions and different functions have different extent of canalization. This provides us a well-defined way to classify all Boolean functions on $n$ variables. 

Returning to our onion analogy, now we can think of all Boolean functions as onions. For each Boolean function, we can try to peel off its variables as we did for nested canalizing functions. We will have to stop once we get to a non-canalizing function. In this sense, nested canalizing functions would be the `best' onions since we can peel off all the variables and non-canalizing would be the `worst'. The $k$-canalizing functions would be those for which one can be peeled off at least $k$ variables. Though a unique core function $f_C=g$ only exists when $g$ is non-constant, we will soon see how every Boolean function, whether or not it has a core function, has a unique \emph{core polynomial} that extends the notion of a core function.

\begin{ex}
The Boolean function $f(x,y,z,w)=xy(z+w)$ has canalizing depth $2$ and core function $f_C=z+w$.
\end{ex}

\begin{rem}
  In our framework, if we consider the set of all Boolean functions on $n$ variables, then:
  \begin{itemize}
  \item The canalizing depth of a $k$-canalizing function is at least $k$.
  \item A non-canalizing function has canalizing depth $0$, and if it is non-constant, then its core function is itself.
  \item Every Boolean function is $0$-canalizing.
  \item The $1$-canalizing functions are precisely the canalizing functions.
  \item The $n$-canalizing functions are precisely the nested canalizing functions.
  \item If a function $f$ has canalizing depth $k$ and the resulting $g$ is constant, then $f$ has $n-k$ fictitious variables, and it is a nested canalizing function on its $k$ essential variables.
  \end{itemize}
\end{rem}

\section{Characterizations of $k$-Canalizing Functions}\label{sec:layers}

\subsection{Polynomial Form of $k$-Canalizing Functions}

It is well-known \cite{lidl1996encyclopedia} that any Boolean function $f$ can be uniquely expressed as a square-free polynomial, called its \emph{algebraic normal form}. Equivalently, the set of Boolean functions on $n$ variables is isomorphic to the quotient ring $R:=\F_2[x_1,\ldots,x_n]/I$, where $I=\langle x_i^2-x_i: 1\leq i \leq n \rangle$. Henceforth in this section, when we speak of Boolean polynomials, we assume they are square-free. Additionally, we will define $\hat{x}_i:=(x_1,\dots,x_{i-1},x_{i+1},\dots,x_n)$ for notational convenience. In this section, we will extend work on NCFs from \cite{li2013boolean} to general $k$-canalizing Boolean functions.

\begin{lem}\label{lem:canalizing}
A Boolean function $f(x_1,\ldots,x_n)$ is canalizing in variable $x_i$, for some $1\leq i\leq n$, with input $a_i$ and output $b_i$, if and only if
\[
f=(x_i+a_i)g(\hat{x}_i)+b_i\,,
\]
for some polynomial $g\not\equiv 0$.
\end{lem} 
\begin{proof}
  Suppose $f$ is canalizing in $x_i$. Write $f$ in its algebraic normal and factor it as 
  \[
  f=x_i\,q(\hat{x}_i)+r(\hat{x}_i)\,,
  \]
  where $q$ and $r$ are the quotient and remainder of $f$ when divided by $x_i$. Note that $b_i=a_iq(\hat{x}_i)+r(\hat{x}_i)$, and since $a_i+a_i=0$ in $\F_2$,
  \[
  f=(x_i+a_i)q(\hat{x}_i)+\big[r(\hat{x}_i)+a_iq(\hat{x}_i)\big]
  =(x_i+a_i)q(\hat{x}_i)+b_i\,.
  \]
  The function $g(\hat{x}_i):=q(\hat{x}_i)$ is nonzero because $f$ is essential in $x_i$. This establishes necessity, and sufficiency is obvious.
\end{proof}
By applying the above lemma recursively, we get the following theorem.
\begin{thm}\label{thm:k-canalizing}
A Boolean function $f(x_1,\ldots,x_n)$ is $k$-canalizing, with respect to permutation $\sigma\in\SSS_n$, inputs $a_i$ and outputs $b_i$, for $1\leq i \leq k$, if and only if it has the polynomial form
\[
f(x_1,\dots,x_n)=(x_{\sigma(1)}+a_1)g(\hat{x}_i)+b_1\,,
\]
where 
\[
g(\hat{x}_i)=(x_{\sigma(2)}+a_2)\Big[\ldots\big[(x_{\sigma(k-1)}+a_{k-1})
[(x_{\sigma(k)}+a_k)\bar{g}+\Delta b_{k-1}]+\Delta b_{k-2}\big]\ldots \Big]+\Delta b_1
\]
for some polynomial $\bar{g}=\bar{g}(x_{\sigma(k+1)},\ldots,x_{\sigma(n)})\not\equiv 0$, where $\Delta b_i:=b_{i+1}-b_i=b_{i+1}+b_i$. $\hfill\Box$
\end{thm}

\subsection{Dominance Layers of Boolean Functions}

One weakness of Theorem \ref{thm:k-canalizing} is that given a Boolean function $f$, the representation of $f$ into the above form, even when $k$ is exactly the canalizing depth, is not unique. In a $k$-canalizing function, some variables are ``more dominant'' than others. We will classify all variables of a Boolean function into different layers according to the extent of their dominance, extending work from \cite{li2013boolean} from NCFs to general Boolean functions. The ``most dominant'' variables will be precisely those that are canalizing. Recall that we are always working in the quotient ring $R=\F_2[x_1,\dots,x_n]/I$, though at times it is helpful to consider the algebraic normal form of a polynomial as an element of $\F_2[x_1,\dots,x_n]$.
\begin{defn}
A Boolean function $M(x_1,\ldots,x_m)$ is an \emph{extended monomial} in variables $x_1,\ldots,x_m$ if 
\[
M(x_1,\ldots,x_m)=\prod_{i=1}^m (x_i+a_i),
\]
where $a_i \in \F_2$ for each $i=1,\ldots,m$. 
\end{defn}

An extended monomial in $R$ is an extended monomial of a subset of $\{x_1,\dots,x_n\}$. In other words, it is simply a product $\prod_{i=1}^n y_i$, where each $y_i$ is either $x_i$, $\overline{x_i}$, or $1$. Using extended monomials, we can refine Theorem \ref{thm:k-canalizing} to obtain a unique \emph{extended monomial form} of any Boolean function.
\begin{prop}\label{prop:dominance}
Given a Boolean function $f(x_1,\ldots,x_n)$, all variables are canalizing if and only if $f=M(x_1,\ldots,x_n)+b$, where $M$ is an extended monomial in all variables.
\end{prop}
\begin{proof}
Suppose all $n$ variables are canalizing if $f$, and so $f$ is essential in every variable. Since $x_1$ is canalizing, Lemma~\ref{lem:canalizing} says that $f=(x_1+a_1)g(\hat{x}_1)+b$ for some $a_1,b\in\F_2$, and $g\not\equiv0$. In particular, this means that $(x_1+a_1)\mid (f+b)$ in $\F_2[x_1,\dots,x_n]$. Since $x_2$ is also canalizing, $f(x_1,a_2,\ldots,x_n)\equiv b'$ for some $a_2$ and $b'$. Plugging in $x_1=a_1$ yields $f(a_1,a_2,x_3,\ldots,x_n)\equiv b=b'$, and so
\[
(x_2+a_2)\mid (f+b)=(x_1+a_1)g(x_2,\ldots,x_n)\,. 
\]
Since $x_1+a_1$ and $x_2+a_2$ are co-prime, we get $(x_2+a_2)\mid g(x_2,\ldots,x_n)$. Note that $g(\hat{x}_1)\not\equiv 0$, hence, we have $g(\hat{x}_1)=(x_2+a_2)g'(x_3,\ldots,x_n)$ where $g'(x_3,\ldots,x_n)\not\equiv 0$. Thus we have $f=(x_1+a_1)(x_2+a_2)g'(x_3,\ldots,x_n)+b$. Necessity of the proposition now follows from induction, and sufficiency is obvious.
\end{proof}

We are now ready to prove the main result of this section. This is a generalized version of Theorem 4.2 in \cite{li2013boolean}. We will obtain a new \emph{extended monomial form} of a Boolean function $f$ by induction. In this form, all variables will be classified into different layers according to their dominance. The canalizing variables are the \emph{most dominant} variables. Thus, a Boolean function may have one, none, or many ``most dominant'' variables. As in \cite{li2013boolean}, variables in the same layer will have the same level of dominance, with the variables in the outer layers being ``more dominant'' than those in the inner layers. 

\begin{thm}\label{thm:layers}
Every Boolean function $f(x_1,\ldots,x_n)\not\equiv 0$ can be uniquely written as
\begin{equation}\label{eqn:layers}
f(x_1,\ldots,x_n)=M_1(M_2(\cdots(M_{r-1}(M_rp_C+1)+1)\cdots)+1)+b,
\end{equation}
where each $M_i=\prod_{j=1}^{k_i}(x_{i_j}+a_{i_j})$ is a nonconstant extended monomial, $p_C\not\equiv 0$ is the \emph{core polynomial} of $f$, and $k=\sum k_i$ is the canalizing depth. Each $x_i$ appears in exactly one of $\{M_1,\dots,M_r,p_C\}$, and the only restrictions on Eq.~\eqref{eqn:layers} are the following ``exceptional cases'':
\begin{enumerate}[(i)]
\item If $p_C\equiv 1$ and $r\neq 1$, then $k_r\geq 2$;
\item If $p_C\equiv 1$ and $r=1$ and $k_1=1$, then $b=0$;
\end{enumerate}
When $f$ is a non-canalizing function, we simply have $p_C=f$.
\end{thm}

Before we prove Theorem~\ref{thm:layers}, we will define some terms and 
 examine a few details, such as the subtle difference between the core function and core polynomial, and the ``exceptional cases'', by simple examples. This should help elucidate the more technical parts of the proof.

\begin{defn}
A Boolean function $f$ written in its unique form from Eq.~\eqref{eqn:layers} is said to be in \emph{standard monomial form}, and $r$ is its \emph{layer number}. The $i^{\rm th}$ \emph{dominance layer} of $f$, denoted $L_i$, is the set of essential variables of $M_i$. The set of essential variables of $p_C$ is
denoted $L_\infty$, and these are called the \emph{recessive variables} of $f$.
\end{defn}

As we will see, when $f$ has a core function $f_C$, its core polynomial is either $p_C=f_C$ or $p_C=f_C+1$. When the number of ``$+1$''s that appear in Eq.~\eqref{eqn:layers}, possibly including $b$, is even, we have $p_C=f_C$. Otherwise, we have $p_C=f_C+1$. When a Boolean function $f$ with canalizing depth $k>0$ fails to have a core function, i.e., the remaining function is either $g\equiv 0$ or $g\equiv 1$, then $f$ is in fact a nested canalizing function on $k$ variables, and its core polynomial is simply $p_C=1$.

Finally, we will examine the two ``exceptional cases''. Both of these are necessary to avoid double-counting certain functions and ensure uniqueness, as claimed in Theorem~\ref{thm:layers}.

\begin{enumerate}[(i)]
\item \emph{If $p_C\equiv 1$ and $r\neq 1$}. In this case, if $k_r=1$, that is $M_r=x_i$ or $\overline{x_i}$, for some $i$. In either case, this innermost layer can be ``absorbed'' into the extended monomial $M_{r-1}$. For example, if $M_r=x_i$, then the inner two layers are
\[
M_{r-1}(M_r+1)+1=M_{r-1}(x_i+1)+1=(x_i+1)\prod_{j=1}^{k_{r-1}}(x_{i_j}+a_{i_j})+1,
=\hat{M}_{r-1}+1\,,
\]
where $\hat{M}_{r-1}=\overline{x_i}M_{r-1}$ is an extended monomial. Thus, in this case we may assume that the innermost layer has at least two essential variables, hence $k_r\geq 2$.

\item If $p_C\equiv 1$ and $r=1$ and $k_1=1$, then for some $i$, either $f=x_i+b$, or $f=\overline{x_i}+b$. Clearly, there are only two such functions, either $f=x_i$ or $f=\overline{x_i}$, and so allowing both $b=0$ and $b=1$ would double-count these. Thus, we may assume that $b=0$. 
\end{enumerate}

\begin{proof}[Proof of Theorem~\ref{thm:layers}]
For any non-canalizing function $f\not\equiv 0$, $f=p_C$ and the uniqueness is obvious.

When $f$ is canalizing, we induct on $n$. When $n=1$, there are $2$ canalizing functions, namely $x=(x)1$ and $x+1=(x+1)1$, both satisfying Eq.~\eqref{eqn:layers}. For the these $2$ functions, since $p_C\equiv1$, $r=1$ and $k_1=1$, we must have $b=0$, so the previous representation is also unique. 

When $n=2$, there are $12$ canalizing functions, $4$ of which are essential in $1$ variable, and thus can be uniquely written as in Eq.~\eqref{eqn:layers}. Now let us consider the $8$ canalizing functions that are essential in $2$ variables. It is easy to check for all these, both variables $x_1$ and $x_2$ are canalizing. Then by Proposition~\ref{prop:dominance}, all of them are of the form
\[
(x_1+a_1)(x_2+a_2)+b=M_1p_C+b\,,
\]
where $M_1=(x_1+a_1)(x_2+a_2)$ and $p_C\equiv1$. In this case, we have $r=1$ and $k_1=2$. Note that when $p_C\equiv1$, the innermost layer must have at least two essential variables, so uniqueness holds. We have proved that Eq.~\eqref{eqn:layers} holds for $n=1$ and $n=2$.

Assume now that Eq.~\eqref{eqn:layers} is true for any canalizing function that is essential in at most $n-1$ variables. Consider a canalizing function $f(x_1,\ldots,x_n)$. Suppose that $x_{1_j}$ for each $j=1,\ldots,k_1$ are all canalizing
in $f$. With the same argument as in Proposition \ref{prop:dominance}, we get $f=M_1g+b$, where $M_1=(x_{1_1}+a_{1_1})\cdots(x_{1_{k_1}}+a_{1_{k_1}})$ and $g\not\equiv 0$. If $g$ is non-canalizing, then Eq.~\eqref{eqn:layers} holds with $p_C=g$ and $r=1$. If $g$ is canalizing, then it is a canalizing function that is essential in at most $n-k_1<n-1$ variables. By our induction hypothesis, it can be uniquely written as
\[
g=M_2(M_3(\cdots(M_{r-1}(M_rp_C+1)+1)\cdots)+1)+b'\,.
\]
Note that $b'$ must be $1$, otherwise all variables in $M_2$ will also be most dominant variables of $f$. This completes the proof.
\end{proof}

\begin{rem}
For any Boolean function $f$:
\begin{enumerate}[(i)]
  \item Variables in two consecutive layers have different canalized outputs.
  \item $L_1$ consists of all the most dominant variables (canalizing variables) of $f$.
\end{enumerate}
\end{rem}

Let us return to our onion analogy, where previously we were peeling off one variable at a time. Furthermore, imagine that each individual variable layer is white if the canalized output $a_i=0$, and black if $a_i=1$. Thus, we can think of an extended monomial layer $L_i$ as a maximal block of variable layers of the same color. We can ``peel off'' an entire $L_i$ at once by plugging in the non-canalizing input $x_{i_j}=\overline{a_{i_j}}$ for each variable in $L_i$. In other words, we can peel off all black layers, then all white layers, then all black layers, and so on. Moreover, we can read off the colors directly off of the function if it is written in the form of Eq.~\eqref{eqn:pncf}. However, recall that this form of a $k$-canalizing function, where $g$ is non-canalizing, is not unique. By Theorem \ref{thm:layers}, the order of consecutive variables, $x_{\sigma(i)}$ and $x_{\sigma(i+1)}$, can be transposed if and only if they are in the same $L_j$. Based on this property, we can enumerate Boolean functions on $n$ variables with canalizing depth $k$. Roughly speaking, we will do this by counting the number of different layer structures, and then counting the number of (non-canalizing) core polynomials. This last set is just the complement of the set of canalizing functions on those variables, which were enumerated in \cite{just2004number}.

\begin{ex}
The Boolean function $f(x_1,\dots,x_7)=x_1\overline{x_2}(x_3x_4(x_5+x_6+x_7+1)+1)$ has canalizing depth $4$. With respect to the permutation $\sigma=1,2,3,4$, its canalizing inputs are $(a_i)_{i=1}^4=(0,1,0,0)$, outputs $(b_i)_{i=1}^4=(0,0,1,1)$ and the core polynomial is $p_C=x_5+x_6+x_7$.
\end{ex}

\section{Enumeration of $k$-canalizing functions}\label{sec:enumeration}

Let $B(n,k)$ be the number of Boolean functions on $n$ variables with canalizing depth exactly $k$. Exact formulas are known for $B(n,k)$ in a few special cases. The number of nested canalizing functions is $B(n,n)$. A recurrence for this was independently derived in the 1970s by engineers studying unate cascade functions  \cite{bender1978asymptotic,sasao1979number}, and then a closed formula was found by mathematicians studying NCFs \cite{li2013boolean}. The quantity $B(4,k)$ was recently computed in \cite{ray2014analysis}. In this section, we will present a general formula for $B(n,k)$. 

Theorem \ref{thm:layers} indicates that we can construct a Boolean functions with canalizing depth $k$ by adding extended monomial layers to a non-canalizing function on $n-k$ variables. Moreover, the complement of the set of non-canalizing functions are the canalizing functions. Hence, let us begin with a formula for $C_n$, the number of canalizing functions on $n$ variables. This result was derived in \cite{just2004number} using a probabilistic method. We will include an alternative combinatorial proof using the \emph{truth table} of a Boolean function $f$. This is the length-$2^n$ vector $(f(x_i))_i$, given some fixed ordering of the elements of $\F_2^n$.
\begin{lem}
The number $C_n$ of canalizing Boolean functions on $n\geq 0$ variables is
\[
C_n=2((-1)^n-n-1)+\sum_{k=1}^n (-1)^{k+1}{n \choose k}2^{k+1}2^{2^{n-k}}.
\]
\end{lem}
\begin{proof}
We wish to count the number of Boolean functions that are canalizing in at least $1$ variable. We can construct a truth table of a Boolean function that is canalizing in at least $k$ variables by doing the following. First, pick $k$ variables to be canalizing; there are ${n \choose k}$ ways to do this. Next, pick the canalizing input for each canalizing variable; there are $2^k$ ways to do that. Then, fill out the entries in the truth table of these canalizing inputs with the same canalized output; there are $2$ ways to do that. The remaining table has $2^{n-k}$ entries, so there are $2^{2^{n-k}}-1$ ways to fill it out such that the corresponding function is non-constant. By inclusion-exclusion, we have $\sum_{k=1}^n (-1)^{k+1}{n \choose k}2^{k+1}(2^{2^{n-k}}-1)$. Note that in this process, there are $2n$ functions of the form $x_i+a_i$, each being counted exactly twice, since we can pick either input as canalizing input. Therefore we have 
\[
C_n=\sum_{k=1}^n (-1)^{k+1}{n \choose k}2^{k+1}(2^{2^{n-k}}-1)-2n=2((-1)^n-n-1)+\sum_{k=1}^n (-1)^{k+1}{n \choose k}2^{k+1}2^{2^{n-k}}.
\]
\end{proof}
As examples, one can check that $C_n=0,2,12,118,3512,\dots$ for $n=0,1,2,3,4,\dots$. This is consistent with the results in \cite{just2004number}, though it should be noted that all numbers differ by $2$ because we do not consider the constant functions to be canalizing. 

Recall that there are $2^{2^n}$ Boolean functions on $n$ variables. Since the non-canalizing functions are the complement of the set of canalizing functions, the following is immediate. 
\begin{cor}
The number $B^*(n,0)$ of non-constant core polynomials on $n$ variables is
\[
B^*(n,0)=B(n,0)-2=(2^{2^n}-C_n)-2=2^{2^{n}}-2((-1)^n-n)+\sum_{k=1}^n (-1)^{k}\binom{n}{k}2^{k+1}2^{2^{n-k}}.
\]
\end{cor}

One can check that $B^*(n,0)=0,0,2,136,62022,\dots$, for $n=0,1,2,3,4,\dots$.

Before we derive the general formula for $B(n,k)$, let us first look at the special case when $k=n$. This was computed in \cite{li2013boolean}, but we include a self-contained proof. Recall that a \emph{composition of $n$} is a sequence $k_1,\dots,k_r$ of non-empty integers such that $k_1+\cdots+k_r=n$. By Theorem~\ref{thm:layers}, the standard extended monomial form of a Boolean function with canalizing depth $k$ involves a size-$r$ composition of $k$ with the additional property that $k_r\geq 2$. 
 
\begin{lem}\label{li2013}
For $n\geq 2$, the number $B(n,n)$ of nested canalizing functions on $n$ variables is given by:
\begin{equation}\label{eqn:B(n,n)}
B(n,n)=2^{n+1} \sum_{\!\!r=1}^{n-1}\sum_{\substack{k_1+\ldots+k_r=n \\ k_i\geq 1,\;k_r\geq 2}} {n \choose k_1,\ldots,k_r}\,,
\end{equation}
where ${n \choose k_1,\ldots,k_r}=\frac{n!}{k_1!k_2!\ldots k_r!}$. 
\end{lem}
\begin{proof}
If a Boolean function is nested canalizing on $n$ variables, then by Theorem \ref{thm:layers}, we know its core polynomial must be $p_C=1$. Let us first fix the layer number $r$. Then for each choice of $k_1,\ldots,k_r$, with $k_1+\ldots+k_r=n$, $k_i\geq 1$ and $k_r\geq 2$, there are ${n\choose k_1,\ldots,k_r}$ different ways to assign $n$ variables to these $r$ layers. For each variable $x_j$, we can pick either $x_j$ or $x_j+1$ to be in its corresponding extended monomial. Note that we also have $2$ choices for $b$. So the number of nested canalizing functions on $n$ variables with exactly $r$ layers is given by:
\[
2^{n+1}\sum_{\substack{k_1+\ldots+k_r=n\\k_i\geq 1,\;k_r\geq 2}} {n \choose k_1,\ldots,k_r}\,.
\]
Then by summing over all possible layer numbers $r$, for $1\leq r\leq n-1$, we get the formula in Eq.~\eqref{eqn:B(n,n)} for $B(n,n)$.
\end{proof}
According to our definition, $B(1,1)=2$. One also can check that $B(2,2)=8$, $B(3,3)=64$, $B(4,4)=736$, $\ldots$.

Now we are ready to derive the general formula for $B(n,k)$.
\begin{thm}
The number $B(n,k)$ of Boolean functions on $n$ variables with canalizing depth $k$, for $1\leq k\leq n$, is
\[
B(n,k)={n\choose k}\left[B(k,k)+B^*(n-k,0)\cdot 2^{k+1}\sum\binom{k}{k_1,\ldots,k_r}\right],
\]
where the sum is taken over all compositions of $k$, and the closed from of $B(k,k)$ is given by Lemma~\ref{li2013}.
\end{thm}
\begin{proof}
We can construct a Boolean function $f$ on $n$ variables with canalizing depth $k$ by doing the following. First, pick $k$ variables that are not in the core polynomial $p_C$. There are ${n\choose k}$ different ways to do that. Once we fixed the variables that are not in $p_C$, we need to consider the following two cases:

\emph{Case 1}: $p_C\equiv 1$. Then $f$ is actually a nested canalizing function on these $k$ variables. There are $B(k,k)$ of them in total.

\emph{Case 2}: $p_C\not\equiv 1$. Then $p_C$ is a non-constant core polynomial on $n-k$ variables, so there are $B^*(n-k,0)$ different choices for $p_C$. Using the same argument as in Lemma \ref{li2013}, there are 
\[
2^{k}\sum\binom{k}{k_1,\ldots,k_r}
\]
different ways for those $k$ variables to form the extended monomials in Equation \eqref{eqn:layers}, where the sum is taken over all compositions of $k$. Note that we also have $2$ ways to pick $b$. Therefore, in this case, there are 
\[
B^*(n-k,0)\cdot 2^{k+1}\sum\binom{k}{k_1,\ldots,k_r}
\]
different Boolean functions.

By combining the above two cases, we get the formula for $B(n,k)$.
\end{proof}
\begin{ex}
As previously mentioned, the quantities $B(4,k)$ for $k=0,\dots,4$ were computed in \cite{ray2014analysis}. It is easy to check that these values are consistent with our general formula. There are $2^{2^4}=65536$ Boolean functions on $4$ variables. The number of functions with canalizing depth exactly $k$, for $k=1,2,3,4$ is 
\begin{align*}
B(4,4)&={4\choose 4}(736+0)=736 \\ B(4,3)&={4\choose 3}(64+0)=256 \\
B(4,2)&={4\choose 2}(8+2\cdot 8\cdot 3)=336. \\ B(4,1)&={4\choose 1}(2+136\cdot 4\cdot 1)=2184. 
\end{align*}
Summing these yields the total number of canalizing functions on $4$ variables,
\[
C_4=3512=736+256+336+2184=B(4,4)+B(4,3)+B(4,2)+B(4,1).
\]
Thus, there are $B(4,0)=65536-3512=62024$ non-canalizing functions on four variables, including the two constant functions.
\end{ex}
Note that $k$-canalizing functions are simply Boolean functions with depth at least $k$, therefore we immediately get the following equality.
\begin{cor}
The number of $k_0$-canalizing Boolean functions on $n$ variables, $1\leq k_0\leq n$, is given by:
\[
\sum_{k=k_0}^n B(n,k)=\sum_{k=k_0}^{n}{n\choose k}\left[B(k,k)+B^*(n-k,0)\cdot 2^{k+1}\sum\binom{k}{k_1,\ldots,k_r}\right].
\]
In particular, the canalizing functions are counted by the following identity:
\[
C_n=\sum_{k=1}^n B(n,k)=\sum_{k=1}^{n}{n\choose k}\left[B(k,k)+B^*(n-k,0)\cdot 2^{k+1}\sum\binom{k}{k_1,\ldots,k_r}\right].
\]
In both equations, the last sum is taken over all compositions of $k$.
\end{cor}

\section{Concluding remarks and future work}\label{sec:conclusions}

Canalizing Boolean functions were inspired by structural and dynamic features of biological networks. In this article, we extended results on NCFs and derived a unique extended monomial form of \emph{arbitrary} Boolean functions. This gave us a stratification of the set of $n$-variable Boolean functions by canalizing depth. In particular, this form encapsulates three invariants of Boolean functions: canalizing depth, dominance layer number and the non-canalizing core polynomial. By combining these three invariants, we obtained an explicit formula for the number of Boolean functions on $n$ variables with depth $k$. We also introduced the notion of $k$-canalizing Boolean functions, which we believe to be a promising framework for modeling gene regulatory networks. Our stratification yielded closed formulas for the number of $n$-variable Boolean functions of canalizing depth $k$. Perhaps more valuable than the exact enumerations is the fact that now it is straightforward to derive asymptotics for the number of such functions as $n$ and $k$ grow large. 

In future work, we will investigate well-known Boolean network models and compute the canalizing depth of the proposed functions. We are working on reverse-engineering algorithms that construct Boolean network models from partial data. In particular, how can one find the function with the maximum canalizing depth that fits that data, and whether the set of $k$-canalizing functions in the model space has an inherent algebraic structure. Progress has been made on these problems for general Boolean functions without paying attention to canalizing depth, and for NCFs. For example, for the general reverse-engineering problem, the set of feasible functions (i.e., the ``model space'') is a coset $f+I$ in the polynomial ring $\F_2[x_1,\dots,x_n]$, where $I$ is the ideal of functions that vanish on the data-set; see~\cite{hinkelmann2012inferring}. Can we get more refined results by restriction to $k$-canalizing functions? The set of nested canalizing functions can be parametrized by a union of toric algebraic varieties~\cite{jarrah2007discrete}. It is relatively straightforward to show that the set of $k$-canalizing functions admits a similar parametrization, but it is not clear whether this has any actual utility for modeling. 

Another avenue of current research extends the work in the electrical engineering community on the unate cascade functions. Recall that these are precisely the NCFs, and they are precisely the functions whose binary decision diagrams have minimum average path length, and this can be explicitly computed. Similarly, we can compute the minimum average path length of a binary decision diagram of a $k$-canalizing function. 

Finally, much of the work in this paper should be able to be extended to multi-state (rather than Boolean) functions. As long as $K$ is a finite field, then $n$-variable functions over $K$ are polynomials in the ring $K[x_1,\dots,x_n]$. The definition of an NCF was extended from Boolean to multi-state functions in \cite{murrugarra2012number}, where the authors also enumerated these functions. Some of the proof techniques in this current paper specifically use the fact that $K=\F_2$, and it is not clear how well they would extend to general finite fields. However, there should absolutely be a stratification of multi-state functions by canalizing depth. The problem of enumerating $k$-canalizing multi-state functions seems to be challenging but still within reach. 


\bibliographystyle{alpha}

\begin{thebibliography}{LAM{\etalchar{+}}13}

\bibitem[AO03]{albert2003topology}
R.~Albert and H.G. Othmer.
\newblock The topology of the regulatory interactions predicts the expression
  pattern of the segment polarity genes in drosophila melanogaster.
\newblock {\em J. Theor. Biol.}, 223(1):1--18, 2003.

\bibitem[BB78]{bender1978asymptotic}
E.A. Bender and J.T. Butler.
\newblock Asymptotic aproximations for the number of fanout-free functions.
\newblock {\em IEEE T. Comput.}, 27(12):1180--1183, 1978.

\bibitem[HJ12]{hinkelmann2012inferring}
F.~Hinkelmann and A.S. Jarrah.
\newblock Inferring biologically relevant models: nested canalyzing functions.
\newblock {\em ISRN Biomathematics}, 2012, 2012.

\bibitem[JL07]{jarrah2007discrete}
A.S. Jarrah and R.~Laubenbacher.
\newblock Discrete models of biochemical networks: The toric variety of nested
  canalyzing functions.
\newblock In {\em Algebraic Biology}, pages 15--22. Springer, 2007.

\bibitem[JM13]{jansen2013phase}
K.~Jansen and M.T. Matache.
\newblock Phase transition of boolean networks with partially nested canalizing
  functions.
\newblock {\em Eur. Phys. J. B}, 86(7):1--11, 2013.

\bibitem[JRL07]{jarrah2007nested}
A.S. Jarrah, B.~Raposa, and R.~Laubenbacher.
\newblock Nested canalyzing, unate cascade, and polynomial functions.
\newblock {\em Physica D}, 233(2):167--174, 2007.

\bibitem[JSK04]{just2004number}
W.~Just, I.~Shmulevich, and J.~Konvalina.
\newblock The number and probability of canalizing functions.
\newblock {\em Physica D}, 197(3):211--221, 2004.

\bibitem[Kau69]{kauffman69metabolic}
S.A. Kauffman.
\newblock Metabolic stability and epigenesis in randomly constructed genetic
  nets.
\newblock {\em J. Theor. Biol.}, 22(3):437--467, 1969.

\bibitem[Kau74]{kauffman1974large}
S.~Kauffman.
\newblock The large scale structure and dynamics of gene control circuits: an
  ensemble approach.
\newblock {\em J. Theor. Biol.}, 44(1):167--190, 1974.

\bibitem[KH07]{karlsson2007order}
F.~Karlsson and M.~H{\"o}rnquist.
\newblock Order or chaos in boolean gene networks depends on the mean fraction
  of canalizing functions.
\newblock {\em Physica A}, 384(2):747--757, 2007.

\bibitem[KKBS14]{klotz2014canalizing}
J.G. Klotz, D.~Kracht, M.~Bossert, and S.~Schober.
\newblock Canalizing boolean functions maximize mutual information.
\newblock {\em IEEE T. Inform. Theory}, 60(4):2139--2147, 2014.

\bibitem[KLAL14]{kadelka2014nested}
C.~Kadelka, Y.~Li, J.O. Adeyeye, and R.~Laubenbacher.
\newblock Nested canalizing functions and their networks.
\newblock {\em arXiv:1411.4067}, 2014.

\bibitem[KPST03]{kauffman2003random}
S.~Kauffman, C.~Peterson, B.~Samuelsson, and C.~Troein.
\newblock Random boolean network models and the yeast transcriptional network.
\newblock {\em Proc. Natl. Acad. Sci.}, 100(25):14796--14799, 2003.

\bibitem[KPST04]{kauffman2004genetic}
S.~Kauffman, C.~Peterson, B.~Samuelsson, and C.~Troein.
\newblock Genetic networks with canalyzing boolean rules are always stable.
\newblock {\em Proc. Natl. Acad. Sci.}, 101(49):17102--17107, 2004.

\bibitem[LAM{\etalchar{+}}13]{li2013boolean}
Y.~Li, J.O. Adeyeye, D.~Murrugarra, B.~Aguilar, and R.~Laubenbacher.
\newblock Boolean nested canalizing functions: A comprehensive analysis.
\newblock {\em Theor. Comput. Sci.}, 481:24--36, 2013.

\bibitem[LDM12]{layne2012nested}
L.~Layne, E.~Dimitrova, and M.~Macauley.
\newblock Nested canalyzing depth and network stability.
\newblock {\em Bull. Math. Biol.}, 74(2):422--433, 2012.

\bibitem[LLL{\etalchar{+}}04]{li2004yeast}
F.~Li, T.~Long, Y.~Lu, Q.~Ouyang, and C.~Tang.
\newblock The yeast cell-cycle network is robustly designed.
\newblock {\em Proc. Acad. Natl. Sci.}, 101(14):4781--4786, 2004.

\bibitem[LNC96]{lidl1996encyclopedia}
R.~Lidl, H.~Niederreiter, and P.M. Cohn.
\newblock Encyclopedia of mathematics and its applications 20: Finite fields,
  1996.

\bibitem[MA05]{moreira2005canalizing}
A.A. Moreira and L.A.N. Amaral.
\newblock Canalizing {K}auffman networks: Nonergodicity and its effect on their
  critical behavior.
\newblock {\em Phys. Rev. Lett.}, 94(21):218702, 2005.

\bibitem[ML12]{murrugarra2012number}
D.~Murrugarra and R.~Laubenbacher.
\newblock The number of multistate nested canalyzing functions.
\newblock {\em Physica D}, 241(10):929--938, 2012.

\bibitem[Pei10]{peixoto2010phase}
T.P. Peixoto.
\newblock The phase diagram of random boolean networks with nested canalizing
  functions.
\newblock {\em Euro. Phys. J. B}, 78(2):187--192, 2010.

\bibitem[RDC14]{ray2014analysis}
C.~Ray, J.K. Das, and P.P. Choudhury.
\newblock On analysis and generation of some biologically important boolean
  functions.
\newblock {\em arXiv:1405.2271}, 2014.

\bibitem[RH13]{robeva2013mathematical}
R.~Robeva and T.~Hodge.
\newblock {\em Mathematical concepts and methods in modern biology: using
  modern discrete models}.
\newblock Academic Press, 2013.

\bibitem[SD07]{szejka2007evolution}
A.~Szejka and B.~Drossel.
\newblock Evolution of canalizing boolean networks.
\newblock {\em Euro. Phys. J. B}, 56(4):373--380, 2007.

\bibitem[SK79]{sasao1979number}
T.~Sasao and K.~Kinoshita.
\newblock On the number of fanout-free functions and unate cascade functions.
\newblock {\em IEEE T. Comput.}, 100(1):66--72, 1979.

\bibitem[SK04]{shmulevich2004activities}
I.~Shmulevich and S.A. Kauffman.
\newblock Activities and sensitivities in boolean network models.
\newblock {\em Phys. Rev. Lett.}, 93(4):048701, 2004.

\end{thebibliography}

\newcommand{\etalchar}[1]{$^{#1}$}

\end{document}